\documentclass[reqno]{amsart}

\usepackage[usenames,dvipsnames]{color}
\usepackage[bookmarks,colorlinks,breaklinks]{hyperref}  

\usepackage{doi,url}

\definecolor{dullmagenta}{rgb}{0.6,0,0.5}   
\definecolor{darkblue}{rgb}{0,0,0.4}
\hypersetup{linkcolor=dullmagenta,citecolor=blue,filecolor=dullmagenta,urlcolor=darkblue} 

\usepackage{amssymb,yhmath}
\usepackage{enumerate}

\newcommand{\eq}[1]{\eqref{#1}}

\newtheorem{theorem}{Theorem}[section]
\newtheorem{proposition}[theorem]{Proposition}
\newtheorem{lemma}[theorem]{Lemma}

\newtheorem{definition}[theorem]{Definition}
\newtheorem{remark}[theorem]{Remark}

\numberwithin{equation}{section}

\newenvironment{acknowledgement}{\emph{Acknowledgement.}}

\DeclareMathOperator{\supp}{supp}
\DeclareMathOperator{\tr}{tr}

\DeclareMathOperator{\diam}{diam}

\DeclareMathOperator{\spr}{spr}

\newcommand\R{\mathbb R}
\newcommand\N{\mathbb N}

\newcommand\Z{\mathbb Z}

\renewcommand\P{\mathbb P}
\newcommand\E{\mathbb E}

\renewcommand\L{\mathrm{L}}

\newcommand{\cc}{\mathrm{c}}

\newcommand\di{\mathrm{d}}

\newcommand\what{\widehat}

\newcommand{\bom}{{\boldsymbol{\omega}}}

\newcommand\eps{\varepsilon}

\newcommand\La{\Lambda}
\newcommand{\vphi}{\varphi}

\newcommand\Chi{\raisebox{.2ex}{$\chi$}}

\newcommand{\abs}[1]{\left\lvert #1 \right\rvert}
\newcommand{\norm}[1]{\left\lVert #1 \right\rVert}
\newcommand{\scal}[1]{\left\langle #1 \right\rangle}
\newcommand{\set}[1]{\left\{ #1 \right\}}
\newcommand{\pa}[1]{\left( #1 \right)}

\newcommand{\up}[1]{^{(#1)}}

\newcommand\beq{\begin{equation}}
\newcommand\eeq{\end{equation}}

\newcommand{\qtx}[1]{\quad\text{#1}\quad}
\newcommand{\sqtx}[1]{\; \text{#1} \;}

\begin{document}

\title[Trimmed discrete Schr\"odinger operators]
{Ground state energy of  trimmed discrete Schr\"odinger operators and  localization for trimmed Anderson models}

\author{Alexander Elgart}
\address[A. Elgart]{Department of Mathematics; Virginia Tech; Blacksburg, VA, 24061, USA}
 \email{aelgart@vt.edu}

\author{Abel Klein}
\address[A. Klein]{University of California, Irvine;
Department of Mathematics;
Irvine, CA 92697-3875,  USA}
 \email{aklein@uci.edu}

\thanks{A.E. was  supported in part by the NSF under grant DMS-1210982.}
\thanks{A.K. was  supported in part by the NSF under grant DMS-1001509.}


\maketitle

\begin{abstract}
We consider  discrete  Schr\"odinger operators of the form
$H=-\Delta +V$ on $\ell^2(\Z^d)$, where  $\Delta$ is the discrete Laplacian and $V$ is a bounded potential. Given $\Gamma \subset \Z^d$,  the $\Gamma$-trimming of $H$ is  the  restriction  of $H$ to  $\ell^2(\Z^d\setminus\Gamma)$, denoted by  $H_\Gamma$.  We investigate the dependence of the ground state energy  $E_\Gamma(H)=\inf \sigma (H_\Gamma)$ on $\Gamma$. We show that for  relatively dense proper subsets $\Gamma$ of   $\Z^d$  we always  have
  $E_\Gamma(H)>E_\emptyset(H)$.  We use this lifting of the ground state energy    to establish Wegner estimates and  localization at the bottom of the spectrum for  $\Gamma$-trimmed Anderson models, i.e.,  Anderson models with the random potential supported by the set $\Gamma$.
\end{abstract}

\tableofcontents

\section{Introduction}

We consider  discrete  Schr\"odinger operators of the form
$H=-\Delta +V$ on $\ell^2(\Z^d)$, where  $\Delta$ is the discrete Laplacian, defined by
\beq
\pa{-\Delta \vphi}(x)=  \sum_{\substack{y \in {\Z^d}, \\ \norm{x-y}=1}}  \pa{\vphi(x)-\vphi(y)}= 2d \vphi(x) -  \sum_{\substack{y \in {\Z^d}, \\\norm{x-y}=1}}   \vphi(y),
\eeq
and $V$ is a bounded potential.  Given $\Gamma \subsetneq \Z^d$,  the $\Gamma$-trimming of $H$ is     the restriction $H_\Gamma$ of  $\Chi_{\Gamma^{\cc}} H \Chi_{\Gamma^{\cc} }$ to $\ell^2(\Gamma^{\cc})$, where $\Chi_A$ denotes the characteristic  function of the set $A$ and $A^\cc= \Z^d \setminus A$ for $A\subset \Z^d$.  We focus our attention on   $E_\Gamma(H)=\inf \sigma (H_\Gamma)$,  the ground state energy (or bottom of the spectrum) of  the trimmed   discrete  Schr\"odinger operator $H_\Gamma$.  (Note that with this notation  $H= H_\emptyset$ and $E_\emptyset(H)= \inf \sigma (H)$.) Since $E_\Gamma(H)$  is a nondecreasing function of the set  $\Gamma$, trimming lifts the bottom of the spectrum, that is, $E_\Gamma(H)\ge E_\emptyset(H)$.

We show that for  relatively dense proper subsets $\Gamma$ of   $\Z^d$  we always  have
 strict lifting of the bottom of  the spectrum, i.e.,  $E_\Gamma(H)>E_\emptyset(H)$.
 We use this lifting of the ground state energy    to establish Wegner estimates and  localization at the bottom of the spectrum for  $\Gamma$-trimmed Anderson models, i.e.,  Anderson models with the random potential supported by the set $\Gamma$.

\subsection{The ground state energy of trimmed  discrete  Schr\"odinger operators}
Our motivation comes from  continuous Schr\"odinger operators $H=-\Delta + V$ on $\L^2(\R^d)$, where $\Delta$ is the  Laplacian operator and $V$ is a bounded potential.   Let us consider first the case  $H=-\Delta$ and $\Gamma^\cc$ an open subset of $\R^d$, and let  $\Delta_\Gamma$ be  the Laplacian on $\Gamma^\cc$ with Dirichlet boundary condition.  When $\overline{\Gamma^\cc}$ is compact, the ground state energy $E_\Gamma(-\Delta)$ of $-\Delta_\Gamma$ is  the first eigenvalue $\lambda_\Gamma$ of   $-\Delta_\Gamma$.   The problem of obtaining a lower bound for the first eigenvalue of the Dirichlet Laplacian on a compact Riemannian manifold has been  intensively studied in Geometric Analysis, and it is given by Cheeger's inequality  \cite{Ch}:
$
\lambda_\Gamma\ge \frac { \beta(\Gamma)^2}{4},
$
  where  $ \beta(\Gamma)$ is Cheeger's isoperimetric constant for the  set $\Gamma^\cc$.
 It is known that $ \beta(\Gamma)>0$ if   $\overline{\Gamma^\cc}$ is compact,    but for noncompact sets $\overline{\Gamma^\cc}$ the Cheeger  isoperimetric constant may be zero.

 Cheeger's inequality has been extended to the discrete case \cite{Dod,LS}, where clearly  $ \beta(\Gamma)>0$ if  $\Gamma^\cc$ is a finite set. But it is not difficult to see that
 $ \beta(\Gamma)=0$ if we can find a  sequence of  boxes in $\Z^d$, $\Lambda_{K_n}(x_n)$  ($\Lambda_K(x)$ is the box of side $K\in \N$ centered at $x\in \Z^d$), such that $\lim_{n\to \infty} \frac{\abs{\Gamma\cap\Lambda_{K_n}(x_n)}}{\abs{\Gamma^{\cc}\cap\Lambda_{K_n}(x_n)}}=0$.   This lead us to consider relatively dense  subsets $\Gamma$ of   $\Z^d$, for which we show $ \beta(\Gamma)>0$.  (See Section \ref{sec:isoperim};  $\beta(\Gamma)$ is defined in \eq{defbetaGamma}.)

The addition of a potential  $V$  breaks down  Cheeger's argument. Indeed, in general  flat functions  are no longer good approximants for the low-lying eigenvectors of $H=-\Delta + V$.  For example,  let $H_\lambda = -\Delta +\lambda V$, where $V$ is a periodic potential  whose average over a fundamental cell is equal to zero.  Then $E_\emptyset(H_\lambda)<0$ for all $\lambda >0$  \cite[Theorem~1]{GGS} (the result there is proven for the continuum, but it is easy to see that  holds in the discrete case as well), but it can be shown that $\beta_\lambda(\emptyset)=0$ for $\lambda$ small, where  $\beta_\lambda(\emptyset)$ is the Cheeger constant for $H_\lambda$.
 Another striking counterexample can be constructed by taking $H_\lambda = -\Delta +\lambda V$ with  $V= - \Chi_{\{0\}}$,  a negative  rank one perturbation to  $-\Delta$, and  $\lambda>0$.   It is  well known that  in this case  $E_\emptyset(H_\lambda)<0$ for all $\lambda >0$, while it is easy to see that  $\beta_\lambda(\emptyset)=0$ for $\lambda\le 2d$.

For continuous Schr\"odinger operators the  bound $E_{\Gamma}(H)>E_{\emptyset}(H)$ can be  established in the presence of an arbitrary bounded potential  using the unique continuation principle   \cite{Kuc,RV}.   Unfortunately,  discrete Schr\"odinger operators do not satisfy  a unique continuation principle. It turns out, however, that {\em the  ground state} of a discrete Schr\"odinger operator $H$ enjoys a similar property, which suffices to establish the desired result.

It is intuitively clear that the Schr\"odinger operator $ H_\Gamma$ is, in a suitable sense,    the limit  of the Schr\"odinger operators $ H_\Gamma(t)=H + t \Chi_\Gamma$  on $\ell^2(\Z^d)$ as $t\rightarrow\infty$.  This is the motivation behind Theorem \ref{thmEinftyEGamma}, where we  obtain a lower bound for $E_{\Gamma}(H)-E_{\emptyset}(H)$ as the limit of lower bounds for $E_{\Gamma}(H,t)-E_{\emptyset}(H)$, where   $E_\Gamma(H,t)= E_\emptyset (H(t))$.    Note that  $E_\Gamma(H,t) $ is nondecreasing in $t$, so   $E_\Gamma(H,\infty):=\lim_{t\to \infty }E_\Gamma(H,t) =\sup_{t\ge 0}E_\Gamma(H,t)$, and it follows from the min-max principle that 
$E_\Gamma(H,t) \le  E_\Gamma(H)$ for all $t\ge 0$, so 
\beq
E_{\emptyset}(H) \le E_\Gamma(H,t)\le  E_\Gamma(H,\infty) \le E_\Gamma(H).
\eeq

Before stating our results, we introduce some additional notation. A  bounded potential $V$ is given  by multiplication by a
 a function $V\colon \Z^d \to \R$ with $V_\infty=\norm{V}_\infty< \infty$.  We set $V_+= \max\set{V,0}$ and
$V_- = - \min\set{V,0}$; note that   $V= V_+ -  V_-$, $V_\pm \ge 0$,  and $V_+ V_-=0$.  
We  define the spread of the bounded potential  $V$  by
 \[\spr (V)= \sup_{x\in \Z^d} V(x) -  \inf_{x\in \Z^d} V(x)\in [0,\infty).\]  
 We  also introduce the following notation:
 \begin{gather} \notag
 Y_{d,V}=2d+1  +{\spr}(V) ,\\  \notag
 \delta_\Gamma(H)=  E_\Gamma(H) - E_{\emptyset}(H),\quad \delta_\Gamma(H,t)=  E_\Gamma(H,t)- E_{\emptyset}(H).
 \end{gather}

\begin{theorem}\label{thmEinftyEGamma} Let   $H=-\Delta +V$ be a Schr\"odinger operator on $\ell^2(\Z^d)$ and $\Gamma \subsetneq \Z^d$. Then
\beq
 E_\Gamma(H,\infty)=E_\Gamma(H).
 \eeq 
 Moreover,
we have 
 \begin{align}\label{Etest33}
2d +  \spr(V)\ge \delta_\Gamma(H,t)\ge \frac {t }{t +6d +2\spr (V)}\,\delta_\Gamma(H) \qtx{for all} t\ge 0.
\end{align}
It follows that  ${E}_\Gamma(H)> E_{\emptyset}(H) $ if and only if $E_\Gamma(H,t)>E_{\emptyset}(H)$ for all $t>0$.
\end{theorem}

Theorem~\ref{thmEinftyEGamma} is proven in Section~\ref{secbottom}. Note that  once  we have a lower bound for $\delta_\Gamma(H) $, as in Theorem~\ref{thmEGammaInt},  \eq{Etest33} (we may use the sharper \eq{eq:intermbnd})  provides lower bounds for $\delta_\Gamma(H,t)$ for all $t> 0$.

Given $x \in \Z^d$ and $L>0$, we set
\beq\notag
\Lambda_L(x)=\set{y \in \Z^d: \; \norm{y-x}_\infty \le\tfrac L 2} \sqtx{and}\Lambda\up{0}_L(x)=\set{y \in \Z^d: \; \norm{y-x}_\infty <\tfrac L 2};
\eeq
note that \
$
\Lambda_L(x) = \Lambda\up{0}_L(x) \quad \Longleftrightarrow \quad L\notin  2 \N$.

Given   $K \in \N$ we have $\abs{\Lambda_K(x)}= K_\ast^d$, where $K_\ast=\begin{cases}    K \qtx{if $K$ is odd} \\
K+1 \qtx{if $K$  is even}
\end{cases}$.
Moreover  $\Lambda_K(x) = \Lambda\up{0}_K(x)$ if and only if $K$ is odd.

\begin{definition}   A   set $\Gamma \subset \Z^d$ is \emph{$(K,Q)$-relatively dense}, where  $K,Q \in\N$, if
\beq
\abs{\Gamma \cap \Lambda\up{0}_K(\zeta)} \ge Q \qtx{for all} \zeta \in  K\Z^d.
 \eeq
\end{definition}

By a relatively dense subset $\Gamma \subset \Z^d$ we will always mean a $(K,Q)$-relatively dense set $\Gamma$ for some appropriate  $K,Q  \in\N$.  Note that we must have  $Q \le  \abs{\Lambda\up{0}_K}\le K^d$, and  that
$K\ge 2$ unless $\Gamma = \Z^d$.

\begin{theorem}\label{thmEGammaInt} Let $\Gamma \subsetneq \Z^d$ be  $(K,Q)$-relatively dense, and let $H=-\Delta +V$ on $\ell^2(\Z^d)$, where $V$ is a bounded potential. Then
\beq\label{delgambndInt}
 \delta_\Gamma(H,t) \ge   \frac Q{\pa{ 2dK-1}{Y_{d,V}^{2dK-1}}}\pa{1-\pa{ \frac {Y_{d,V}}{Y_{d,V}+t}}^{2dK-1}}\qtx{for all} t\ge 0.
 \eeq
  As a consequence, we get
\beq\label{EGammaboundVInt}
 \delta_\Gamma(H)\ge  \frac Q{\pa{ 2dK-1}{Y_{d,V}^{2dK-1}}}>0 .
\eeq

In the special case $H=-\Delta$ we can improve the previous bound to
\beq\label{EGammaboundInt}
\delta_\Gamma(-\Delta) = E_{\Gamma}(-\Delta) \ge \frac 1 {4dK_\ast^{2d}}\, .
\eeq
\end{theorem}

We prove  \eq{delgambndInt} from a `quantitative unique continuation principle for ground states' given in Lemma~\ref{lemgraph}.  The lower bound given in \eq{EGammaboundVInt} holds for arbitrary bounded potential $V$; note  that it depends on $V$ only through  ${\spr}(V)$.

The special case 
 \eq{EGammaboundInt} follows from a  Cheeger's inequality.  We remark that $E_\Gamma(-\Delta)$ can also be estimated by an argument of Bourgain and Kenig \cite[Section~4]{BK} (see also \cite[Remark~4.4]{GKloc}). They treated  continuum models, but Rojas-Molina \cite[Section~1.2.5]{R}, \cite[Lemma~2.1]{R2} noted that the argument applies also to the discrete case.  This argument  yields   the bound  $E_{\Gamma}(-\Delta) \ge \frac C {K^{2d+2}}$.

  Theorem~\ref{thmEGammaInt} has a continuum counterpart; in particular we can use Cheeger's inequality to obtain the continuum version of \eqref{EGammaboundInt}. We did not include it in this paper because the continuous version of the general estimate \eqref{delgambndInt} is only marginally better than the estimate in  \cite[Lemma~4.2]{Kuc}, and the  continuous version of \eqref{EGammaboundInt} is similary only marginally better  that what we get with the Bourgain-Kenig argument.

Theorem~\ref{thmEGammaInt}  follows from  Theorems~\ref{thmEGammaV} and \ref{thmEGammaDelta} in Section~\ref{seclbgs}.

\subsection{Wegner estimates and localization for trimmed Anderson models}  If $\zeta \in \Z^d$, we will use the notation   $\Chi_\zeta$ for $\Chi_{\set{\zeta}}$  as a multiplication operator in $\ell^2(\Z^d)$,  but we will write $\delta_\zeta$ instead  when considering $\Chi_{\set{\zeta}}$ as a function in $\ell^2(\Z^d)$.

\subsubsection{Trimmed  Anderson models} 

\begin{definition}  \label{defDilAndHam} Let $\Gamma \subset \Z^d$ be  $(K,Q)$-relatively dense.   A \emph{$\Gamma$-trimmed  Anderson model} is a  discrete random Schr\"odinger  operator on on $\ell^2(\Z^d)$
 of the form
\beq\label{AndH}
H_{\bom,\lambda}: =  H_0+
\lambda V_{\bom} ,
\eeq
where:

\begin{enumerate}
\item  $H_0 = -\Delta + V\up{0}$,  with   $ V\up{0}$ a bounded   (background) potential.
\item  $V_{\bom}$ is the random potential given by
\beq
V_{\bom} :=
\sum_{\zeta \in \Gamma} \omega_\zeta   \Chi_\zeta ,\label{AndVK}
\eeq
where  $\bom=\{ \omega_\zeta \}_{\zeta \in
\Gamma}$ is a family of independent
 random
variables  whose  probability
distributions $\{ \mu_\zeta\}_{\zeta \in
\Gamma}$ are non\--dege\-nerate with
\beq \label{mu}
 \supp \mu_\zeta \subset   [0,M ]  \qtx{for all} \zeta \in \Z^d.
\eeq

\item $\lambda >0$ is the disorder parameter.

\end{enumerate}

\end{definition}

If $\Gamma=\Z^d$, $ V\up{0}=0$, and       $\mu_\zeta=\mu$ for all $\zeta \in \Z^d$, then $H_{\bom,\lambda}$ is the standard   Anderson model.   This model was introduced by Anderson  \cite{And} to study the effect of disorder on electronic states within the suitable energy range.  The  main phenomenon  is localization,  which manifests itself as spectral localization (the spectral measure of $H_{\bom,\lambda}$ is almost surely pure point with  exponential decay of eigenfunctions) and as dynamical localization (non-spreading of wave packets).

 Trimmed Anderson models are  the discrete analogues of the crooked Anderson Hamiltonians introduced in \cite[Definition~1.2]{Kuc}. (By a trimmed Anderson model we will always mean a $\Gamma$-trimmed  Anderson model for some relatively dense subset $\Gamma \subset \Z^d$.)

 The standard  Anderson model with sufficiently regular single site probability distribution $\mu$  was intensively studied during the last two decades; see \cite{A,AM,ASFH,DLS,vDK,FK,FMSS,FS,Klweak,SW,W} and the reviews \cite{GKicmp,Ki,S} for a more exhaustive list of  references. (In this paper we consider only results valid in  arbitrary dimension $d$; the $d=1$ case is special and we will not mention $d=1$ only results.) It exhibits localization in an interval at the bottom of the spectrum for fixed disorder and on the whole real line  for large disorder.   On the other hand, until very recently there had been 
no localization results for  ergodic $\Gamma$- trimmed Anderson models with $\Gamma \ne \Z^d$, say $\Gamma= K\Z^d$ with $K \ge 2$.
  The reason is the lack of a covering condition, i.e., that the support of the random potential is  all of $\Z^d$ with probability one.  Indeed, $\sum_{\zeta  \in K \Z^d}   \Chi_\zeta =\Chi_{\Gamma}$,  and hence
$\sum_{\zeta \in K \Z^d}   \Chi_\zeta \ge c> 0$ if and only if $\Gamma = \Z^d$.   The covering condition has   played  a crucial role in the proofs of Wegner estimates (which are bounds on the regularity of the integrated density of states, first proved by Wegner \cite{Weg} for the standard Anderson model) and  localization for the Anderson model.

This difficulty has been overcome for the continuous analogue of the Anderson model by the use of   the unique continuation principle for continuous Schr\"odinger operators, and  localization at the bottom of the spectrum has been  proved  for continuous Anderson Hamiltonians  \cite{BK,CHK1,CHK2,GKloc}.  These results were further  extended to a larger  class of  continuous  random  Schr\" odinger operators with alloy-type random potentials, including non-ergodic random  Schr\" odinger operators such as Delone-Anderson Hamiltonians \cite{Kuc,R,RV}.

Recently, Rojas-Molina  \cite[Theorem~1.2.6]{R} proved Wegner estimates and localization at the bottom of the spectrum for the  special case of $2\Z^d$-trimmed Anderson models with no background potential, i.e.,  $V\up{0}=0$. She circumvented  the lack of covering condition  using  an argument of Bourgain and Kenig \cite[Section~4]{BK}  as described  in  \cite[Remark~4.4]{GKloc}.   Her approach can be extended for   $\Gamma$-trimmed Anderson models with $\Gamma$ an arbitrary relatively dense subset of $\Z^d$, as long as there is no background potential \cite[Section~2.1]{R2};   the Bourgain-Kenig argument does not appear to be able to incorporate a background potential.
 Cao and Elgart  \cite{CE} showed localization at small disorder below  the bottom of the free  spectrum for a class of  three-dimensional Anderson-like models  without background potential. The random variables in \cite{CE} are supported on the interval $[-1,1]$ which in our context corresponds to a particular type of $\Gamma$-trimmed Anderson models.

Although there is no unique continuation principle for discrete Schr\"odinger operators,   we prove  Wegner estimates and localization at the bottom of the spectrum for $\Gamma$-trimmed Anderson models with  nontrivial background potentials.
We are not aware of any previous results on either Wegner estimates or localization for this class of models.

\subsubsection{The ground state energy}
A trimmed Anderson model $H_{\bom,\lambda}$ is a $K\Z^d$-ergodic random Schr\"odinger operator if and only if 
 $\Gamma=\Gamma + \zeta$ for all $\zeta \in KZ^d$,  $ V\up{0}$ is a periodic potential with period $K$, and       $\mu_\zeta=\mu$ for all $\zeta \in \Gamma$.  In this case its spectrum $\sigma(H_{\bom,\lambda})$ is not random, i.e., it is the same with probability one.  
  In particular, requiring $0=\inf \supp \mu$, we get 
  \beq\label{E0prob}
   E_\emptyset(H_{\bom,\lambda})=E_\emptyset(H_0) \mbox{ with probability one}.
   \eeq

Since  a trimmed  Anderson model $H_{\bom,\lambda}$ is not, in general, an ergodic random operator, its spectrum $\sigma(H_{\bom,\lambda})$ is a random set. We have  $E_\emptyset (H_{\bom,\lambda}) \ge E_\emptyset (H_0)$ for all $\bom\in [0,M]^\Gamma$ and $\lambda>0$.
But even  after imposing $\mu_\zeta=\mu$ for all $\zeta \in \Gamma$ with $0=\inf \supp \mu$ we cannot guarantee \eqref{E0prob}.
For example, take $V\up{0}= - 6d \Chi_{\zeta_{0}}$ for some $\zeta_{0}\in \Gamma$, $\mu$ uniformly distributed on $[0,1]$, and $\lambda > 6d$.  Then $E_\emptyset(H_0)\le \scal{\delta_{\zeta_{0}}, H_0 \delta_{\zeta_{0}}}= - 4d   $, but we clearly have 
$\P \set{E_\emptyset(H_{\bom,\lambda})\ge 0} >0$, so \eq{E0prob} is not true.   But if in addition we require $ V\up{0}$ to be a periodic potential with period $K$,  it  follows that  \eq{E0prob} holds by comparison with the ergodic random operator we obtain by removing the $\Gamma$-trimming, that is, replacing   $\Gamma$ by $\Z^d$.
Actually, \eq{E0prob} holds in a broader context as the following proposition will show. (See also \cite{R2}.)

Given a Schr\"odinger operator $H$ on $\ell^2(\Z^d)$, we define finite volume operators $H\up{\La}= H_{{\La}^{\cc}}$, i.e., the restriction of $\Chi_\La H \Chi_\La$ to $\ell^2(\La)$ where  $\La=\Lambda_L(x)$ is a finite box.  In particular,
given a a trimmed Anderson model  $H_{\bom,\lambda}$,  we define finite volume random operators  $H_{\bom,\lambda}\up{\La}$. We also set
 $S_\La(t) := \max_{\zeta\in \Gamma \cap \La}    S_{\mu_\zeta}(t) $ for $t\ge 0$, where $S_{\mu}(t):=  \sup_{a \in \R} \mu([a,a+t]) $ denotes  the concentration function of the probability measure $\mu$, and let   $S(t) := \sup_{\zeta\in \Gamma }    S_{\mu_\zeta}(t) $ for $t\ge 0$

\begin{proposition} \label{prop:admfam} Let $H_{\bom,\lambda}$ be a $\Gamma$-trimmed  Anderson model with  $\mu_\zeta=\mu$ for all $\zeta \in \Gamma$ with $0=\inf \supp \mu$.  Suppose  for any $\epsilon>0$ there is $L=L(\epsilon)>0$ such that
\beq\label{eq:condprob}
\abs{\set{x\in\Z^d: E_\emptyset(H_{0}\up{\La_L(x)})\le E_\emptyset + \epsilon}}=\infty.
\eeq
Then  $E_\emptyset(H_0)$ is in the essential spectrum of $H_0$ and   $  E_\emptyset(H_{\bom,\lambda})=E_\emptyset(H_0)$ with probability one.
\end{proposition}

The proof is given  in Section~\ref{secbottomsp}.

\subsubsection{ Wegner estimates and localization}
We  prove Wegner estimates and localization for $\Gamma$-trimmed  Anderson models in intervals $[E_\emptyset(H_0),E_1]\subset [E_\emptyset(H_0),E_\Gamma(H_0)) $. 
Note  that we have \eq{E0prob}, and hence almost sure existence of the spectrum in these intervals, that is, 
\beq
\P\set{\sigma(H_{\bom,\lambda}) \cap [E_\emptyset(H_0),E_1]\ne \emptyset}=1 \qtx{for all} E_1 > E_\emptyset(H_0),
\eeq
for the  class of (generally) non-ergodic trimmed  Anderson models given in Proposition~\ref{prop:admfam}.

 \begin{theorem}\label{thmWegbottom} Let $H_{\bom,\lambda}$ be a $\Gamma$-trimmed  Anderson model.    Given  an energy $E_{1}\in (E_\emptyset(H_0), E_\Gamma(H_0))$,   set
 \beq\label{kappabottom}
\kappa = \kappa(H_0,\Gamma,E_1)= \sup_{s>0; \; E_\Gamma(H_0,s) >E_1} \frac {E_\Gamma(H_0,s) -E_{1}}s   >0.
\eeq
Then
 for every box   $\La=\La_L(x_0)$ with $x_0 \in \Z^d$ and $L>0$ we have
 \beq \label{PVPb}
\Chi_{(-\infty,E_{1}]}(H_{\bom,\lambda}\up{\La})\Chi_{\Gamma \cap \La}\Chi_{(-\infty,E_{1}]}(H_{\bom,\lambda}\up{\La}) \ge \kappa \Chi_{(-\infty,E_{1}]}(H_{\bom,\lambda}\up{\La}) \sqtx{for} \bom \in [0,M]^\Gamma,  
\eeq
and for any closed interval $I \subset (-\infty,E_{1}]$ we have
\begin{align}\label{preWegner}
\E\set{ \tr \Chi_{I}(H_{\bom,\lambda}\up{\La}) }   \le  8\kappa^{-1}S_\La(\lambda^{-1}\abs{I})\abs{\Gamma \cap \La}.
    \end{align}
\end{theorem}

\begin{remark} \label{remkappa} It follows from \eqref{delgambndInt} (and its proof) that
\beq\label{eq:bndkappa}
\kappa \ge  \frac{Q} {2dK+1}\, \pa{(1+Z)Y_{d,V\up{0}}} ^{-2dK}
\eeq
where
\beq
Z= \frac {2Kd+1}{2Kd} \pa{ \pa{1-   { (E_1-E_\emptyset(H_0))Q^{-1}(2dK-1)Y_{d,V\up{0}}^{2dK-1}} }^{-\frac 1 { 2dK-1}}   -1}.
\eeq
(See  \eq{kappa1}-\eq{kappa5}.)
\end{remark}

Theorem \ref{thmWegbottom} is proved in Section~\ref{secwegner}.

{ The Wegner type estimate \eqref{preWegner} allows us to establish  localization  for $\Gamma$-trimmed  Anderson models at the bottom of the spectrum. By complete localization on an interval $I$ we mean that for all $E\in I$ there exists $\delta(E)>0$ such that  we can perform the bootstrap multiscale analysis on the interval $(E-\delta(E),E+\delta(E))$, obtaining Anderson and dynamical localization; see \cite{GKboot,GKfinvol,GKsudec}. (Note that by this definition  we always have complete localization in  $(-\infty, E_\emptyset (H_0))$.)

The following theorem show that we always have  localization below $E_\emptyset (H_0)$ at high disorder.

\begin{theorem} \label{comploclambda} Let  $H_{\bom,\lambda}$ be a $\Gamma$-trimmed Anderson model,    and suppose   $S(t)  \le C t^\theta$ for all  $t\ge 0$, where $\theta \in (0,1]$ and $C$ is a constant.  Then, given  $E_{1}\in (E_\emptyset (H_0), E_\Gamma (H_0))$,  there exists $\lambda(E_1) <\infty$  such that $H_{\bom,\lambda}$ exhibits complete localization on the interval $(-\infty, E_1)$ for all $\lambda \ge \lambda(E_1)$.
\end{theorem}

Theorem~\ref{comploclambda} is proved exactly as \cite[Theorem~1.7]{Kuc} using  the Wegner estimate \eq{preWegner}, so we omit the proof.

We also establish localization in an interval at the bottom of the spectrum for fixed disorder.

\begin{theorem} \label{comploclambdafixed} Let  $H_{\bom,\lambda}$ be a $\Gamma$-trimmed Anderson model,   and suppose   $S(t)  \le C t^\theta$ for all  $t\ge 0$, where $\theta \in (0,1]$ and $C$ is a constant. Assume in addition that one of the following hypotheses hold:
\begin{enumerate}
\item $H_{\bom,\lambda}$ is  an ergodic $\Gamma$- trimmed Anderson model.

\item  There is no background potential, that is, $V\up{0}=0$.
\item  The  exponent $\theta$ satisfies $\theta > \frac d 2 $.
\end{enumerate}
Then for all $\lambda>0$ there exists $E_{\lambda }>E_\emptyset (H_0)$ such that $H_{\bom,\lambda}$ exhibits complete localization on the interval $(-\infty, E_\lambda)$. 
\end{theorem}

The proof of this theorem is standard once we have the Wegner estimate \eq{preWegner}. (Thus we will have $E_\lambda < E_\Gamma(H_0)$.)  The  necessary input for starting the multiscale analysis can be verified as follows:
\begin{enumerate}
\item\label{i:lif} If $H_{\bom,\lambda}$ is ergodic, it  has Lifshitz tails \cite{M,Klop99}(the proofs apply also to the discrete case),   and we proceed as in \cite[Proposition~4.3]{GKloc}.

\item If $V\up{0}=0$, we proceed as in \cite[Remark~4.4]{GKloc}; the argument can be adapted to the discrete case as noted in \cite[Theorem~1.2.6]{R}, \cite{R2}.

\item  If $\theta > \frac d 2 $, we employ the same strategy as in \ref{i:lif}, replacing the Lifshistz tails with the  `classical tails' given by  the condition  $\theta > \frac d 2 $ as in  \cite[Proof of Theorem~3$^\prime$]{FK}, \cite[Proof of Theorem~3.11]{KK}.

\end{enumerate}

\section{The ground state energy  of trimmed Schr\"odinger operators}

In this section we prove Theorems~\ref{thmEinftyEGamma} and \ref{thmEGammaInt}.

Given $H=-\Delta +V$ on $\ell^2(\Z^d)$, where $V$ is a bounded potential, we will use the shorthand notation  $E_\Gamma=E_\Gamma(H)$, $E_\Gamma(t)=E_\Gamma(H,t)$, $E_{\emptyset}=E_{\emptyset}(H)$.

\subsection{Equality of the ground state energies}\label{secbottom}
We start by  proving  Theorem~\ref{thmEinftyEGamma}.
\begin{proof}[Proof of Theorem~\ref{thmEinftyEGamma}]

We first obtain a simple upper bound on $\delta_\Gamma(H)$ (hence on $\delta_\Gamma(H,t)$ as well), to be used later on. To this end, note that  $E_\emptyset \ge \inf_{x\in \Z^d} V(x)$, and hence
\beq\label{HGamma <}
H_\Gamma - E_\emptyset \le -\Delta_\Gamma + \spr{V} \le 4d +  \spr{V}.
\eeq
It follows that
\beq\label{HGamma <1}
E_\Gamma  - E_\emptyset \le E_\Gamma(-\Delta) + \spr{V}\le 2d + \spr{V},
\eeq
where we used $E_\Gamma(-\Delta)\le \inf_{x\in \Gamma^{\cc}} \scal{\delta_x, -\Delta \delta_x}=2d$,
that is,
\beq
\delta_\Gamma(H)\le  
\delta_\Gamma(-\Delta ) + \spr(V) \le 2d +  \spr(V).
\eeq

Suppose $E_\Gamma >E_{\emptyset}$, since otherwise there is nothing to prove.    
By replacing $H$ by $H-E_{\emptyset}$, we may assume $E_{\emptyset}=0$, so $\delta_\Gamma(H)=  E_\Gamma$ and
  $\delta_\Gamma(H,t)=  E_\Gamma(t)$.   

 Let  $\nu>0$. Then   $H +\nu \ge \nu$ (recall  $E_{\emptyset}=0$), so $(H+\nu)^{-1}\le \frac 1 \nu$.  It follows that  on $\ell^2(\Gamma)$ we have
\beq
S_\nu=H_{\Gamma^\cc}+\nu-u(H_\Gamma+\nu)^{-1} u^*\ge \nu ,\quad u=\chi_\Gamma \Delta \chi_{\Gamma^{\cc}},
\eeq
since $S_\nu$
is the the Schur complement of $H_\Gamma+ \nu$, and  we have 
\beq
S_\nu^{-1}=\Chi_{\Gamma} (H+\nu)^{-1}  \Chi_{\Gamma} \le \frac 1 \nu.
\eeq
In particular, we conclude that
\beq\label{schurest4}
H_{\Gamma^\cc}\ge u(H_\Gamma+\nu)^{-1} u^*  \qtx{for all} \nu >0.
\eeq

By hypothesis $E_\Gamma>0$, so we take $\eta \in (0,E_\Gamma)$.  Note that for all $\nu>0$ we have
\beq
(H_\Gamma-\eta)^{-1}\le \left(1+\tfrac{\eta + \nu}{E_\Gamma - \eta}\right)(H_\Gamma+\nu)^{-1}.
\eeq
We now  consider the Schur complement $S_{-\eta}(t)$  of $\pa{H_\Gamma(t)}_\Gamma -\eta$, and use \eq{schurest4} and \eq{HGamma <}, getting
\begin{align}\notag
S_{-\eta}(t)&= H_{\Gamma^\cc}+t - \eta-u(H_\Gamma-\eta)^{-1} u^*\\
& \ge  H_{\Gamma^\cc}+t - \eta-
\left(1+\tfrac{\eta + \nu}{E_\Gamma - \eta}\right) u(H_\Gamma+\nu)^{-1} u^*\\
& \ge H_{\Gamma^\cc}+t - \eta- \left(1+\tfrac{\eta + \nu}{E_\Gamma - \eta}\right)H_{\Gamma^\cc}\ge t -\eta -\tfrac{\eta + \nu}{E_\Gamma - \eta}(4d +{\spr} (V)).
\notag\end{align}
Since $\nu>0$ is arbitrary, we obtain
\beq
S_{-\eta}(t)\ge t -\eta -\tfrac{\eta }{E_\Gamma - \eta}(4d  +{\spr} (V)).
\eeq
We conclude that
\beq
S_{-\eta}(t)>0 \sqtx{if} \eta < \frac {t +4d + E_\Gamma +{\spr}(V) -\sqrt{\pa{t +4d + E_\Gamma +{\spr} (V)}^2-4E_\Gamma t}}{2 },
\eeq
so it follows from the Schur complement condition for positive definiteness that 
\begin{align}\notag
E_\Gamma(t)& \ge  \frac {t +4d + E_\Gamma +{\spr} (V) -\sqrt{\pa{t +4d + E_\Gamma +{\spr} (V)}^2-4E_\Gamma t}}{2 }\\  \notag & =
 \frac {2 E_\Gamma t}{t +4d + E_\Gamma +{\spr}
  (V) +\sqrt{\pa{t +4d + E_\Gamma +{\spr} (V)}^2-4E_\Gamma t}}\\  & \ge \frac {E_\Gamma t}{t +4d + E_\Gamma +{\spr}
   (V) } \qtx{for} t>0,  \label{eq:intermbnd}
\end{align}
 Combining with  \eq{HGamma <1} we get 
\beq \label{Etest334}
E_\Gamma(t) \ge \frac {E_\Gamma t}{t +6d  + 2{\spr}(V) }   \qtx{for} t>0,
\eeq
which is \eq{Etest33},
Leting $t \to \infty$ in \eq{Etest33} we get $E_\Gamma(\infty)\ge E_\Gamma$.  SInce $E_\Gamma(\infty)\le E_\Gamma$,
we get $E_\Gamma(\infty)=E_\Gamma$.
\end{proof}

\subsection{Lower bounds on the ground state energy for arbitrary potential}  \label{seclbgs} Theorem~\ref{thmEGammaInt} for arbitrary bounded potential $V$, namely the lower 
 bounds \eqref{delgambndInt}-\eqref{EGammaboundVInt}, follows from the following theorem.

 We recall $Y_{d,V}=2d+1  +{\spr}(V)$ for a  bounded potential $V$.

\begin{theorem}\label{thmEGammaV} Let $\Gamma \subsetneq \Z^d$ be  $(K,Q)$-relatively dense, and let $H=-\Delta +V$ on $\ell^2(\Z^d)$, where $V$ is a bounded potential.
Then  
\beq\label{delgambnd}
 \delta_\Gamma(H,t) \ge  \frac {Q}{ 2dK-1}\pa{\frac {1}{Y_{d,V}^{2dK-1}}-  \frac {1}{\pa{Y_{d,V}+t}^{2dK-1}}}\qtx{for all} t\ge 0.
\eeq
 As a consequence, we get
\beq\label{EGammaboundV}
 \delta_\Gamma(H)\ge  \frac Q{\pa{ 2dK-1}{Y_{d,V}^{2dK-1}}}>0 .
\eeq
\end{theorem}

The proof of the theorem is based on what may be called a quantitative unique continuation principle for ground states, given in the following lemma.

Given  a nonempty connected  subset $B$ of $\Z^d$  and $x,y \in B$,   we let $d_B(x,y)$ denote the graph distance between $x$ and $y$ in $B$, i.,e., the minimal length of a path in $B$ connecting  $x$ and $y$.  We set  $\diam B= \max_{x,y \in B}d_B(x,y)$, the diameter of $B$ in the graph theory sense.  Note that  we always have $d_B(x,y) \ge \norm{x-y}_1$, and $d_B(x,y) =\norm{x-y}_1$ for all $x,y \in B$ if $B=\Z^d$ or  $B=\Lambda_L(x_0)$.  In particular, we have
$\diam \Lambda_L(x_0)\le dL$.

\begin{lemma}\label{lemgraph} Let  $H=-\Delta +V$ on $\ell^2(\Z^d)$, where $V$ is a bounded potential.
  Let $\La=\Lambda_L(x_0)$ be a box in $\Z^d$. Then $E\up{\La}= \inf \sigma{(H_\La)}$ is a simple eigenvalue, and there exists a unique strictly positive  ground state $\psi_g\up{\La}$, i.e.,  there exists a unique $\psi_g\up{\La}\in \ell^2(\La)$ such that $H_\La \psi_g\up{\La} = E\up{\La}\psi_g\up{\La}$,
$\norm{\psi_g\up{\La}}=1$, and $\psi_g\up{\La}(x)>0$ for all $x \in \La$.  Moreover, for all $x \in \La$ and $m \in \N$ we have
\beq\label{graph1}
\psi_g\up{\La}(x) \ge
Y_{d,V}^{-m }\sum_{y \in \La; \ \norm{x-y}_1\le m} \psi_g\up{\La}(y).
\eeq
We also get a uniform lower bound: 
\beq\label{Lalowerbound}
\psi_g\up{\La}(x) \ge
Y_{d,V}^{-dL }\qtx{for all} x \in \La.
\eeq
\end{lemma}

\begin{proof} Without loss of generality we assume $0=\inf_{x\in \Z^d} V(x)$, so $0\le V \le V_\infty={\spr}(V)$ and $ E\up{\La}\ge 0$.

 Note that $\ell^2(\La)$ is a finite-dimensional Hilbert space. The existence of the unique stricty positive ground state follows from the  Perron-Frobenius Theorem. This can be seen as follows.  The  self-adjoint operator   $T=2d+1 +V_{\infty}  -H_\La$ on  $\ell^2(\La)$ is  positivity preserving, i.e.,   $\scal{\delta_x, T \delta_{y}}\ge 0$ for all $x,y \in \La$. Moreover,
\beq
\scal{\delta_x, T^{m} \delta_{y}}\ge 1 \sqtx{for} m\ge \norm{x-y}_1 \sqtx{for  all} x,y\in \La.
\eeq
In particular, recalling $\diam \La \le dL$, we have
\beq
\scal{\delta_x, T^{dL} \delta_{y}}\ge 1 \sqtx{for  all} x,y\in \La.
\eeq

It follows from the Perron-Frobenius Theorem that $\lambda_{\max}= \max \sigma(T)$ is a simple eigenvalue, and there exists a unique $\psi_g\up{\La}\in \ell^2(\La)$ such that $T \psi_g\up{\La} =\lambda_{\max}\psi_g\up{\La}$,
$\norm{\psi_g\up{\La}}=1$, and $\psi_g\up{\La}(x)>0$ for all $x \in \La$.  Clearly, $H_\La \psi_g\up{\La} = E\up{\La}\psi_g\up{\La}$ and
\beq
\lambda_{\max}= 2d+1 +V_{\infty}  -  E\up{\La} \le  2d+1  +V_{\infty}=Y_{d,V}.
\eeq
Moreover,  since $T \psi_g\up{\La} =\lambda_{\max}\psi_g\up{\La}$ and  $\psi_g\up{\La}(x)>0$ for all $x \in \La$, we have for all $x \in \La$ and $m \in \N$  ($\psi_g= \psi_g\up{\La}$)
\beq
\psi_g(x) \ge \lambda_{\max}^{-m }\sum_{y \in \La; \ \norm{x-y}_1\le m} \psi_g(y) ,
\eeq
which yields \eq{graph1}.

To get \eq{Lalowerbound}, just notice that  $1=\norm{\psi_g}_2 \le \norm{\psi_g}_1= \sum_{y \in \La} \psi_g(y)$.
\end{proof}

\begin{proof} [Proof of Theorem~\ref{thmEGammaV}]  Given $\zeta \in  K\Z^d$, fix $ \Gamma_\zeta\subset \Gamma \cap \Lambda\up{0}_K(\zeta)$ such that $\abs{\Gamma_\zeta}=Q$.

Let $R=KJ$ where $J=1,3,5,\ldots$ and consider $\La=\La_R=\La_R(0)$.
Then, by Lemma~\ref{lemgraph}, for all $t\ge 0$ we have that $E\up{\La}(t)= \inf \sigma( H_{\Gamma}(t))$ is a simple isolated eigenvalue with eigenvector  $\psi_{g,t}\up{\La}$  as in Lemma~\ref{lemgraph}, so it follows  that the orthogonal projection  $P_g(t)=\scal {\psi_{g,t}\up{\La}, \ \cdot \  \psi_{g,t}\up{\La}} \psi_{g,t}\up{\La}$ is differentiable in $t$, and
\begin{align}\label{dec1}
\frac \di {\di t} E\up{\La}(t) &= \frac \di {\di t} \tr P_g(t) H(t)\\ \notag  &= \tr P_g(t)\dot H(t)+\tr \dot P_g(t)H(t) = \tr P_g(t)\dot H(t)=\tr P_g(t) \chi_{\Gamma}\\ \notag  & \ge \sum_{\zeta \in  K\Z^d\cap \La}
\scal {\psi_{g,t}\up{\La}, \Chi_{\Gamma_\zeta} \psi_{g,t}\up{\La}}
=\sum_{\zeta \in  K\Z^d\cap \La}\sum_{x \in \Gamma_\zeta}  \pa{\psi_{g,t}\up{\La}(x)} ^2,
\end{align}
where on the second line we have used $\dot P_g=P_g\dot P_g(1-P_g)+(1-P_g)\dot P_g P_g$, cyclicity of the trace, and $P_gH(1-P_g)=0$.

If $x \in \Gamma_\zeta$, it follows from \eq{graph1} that
\beq
\psi_{g,t}\up{\La}(x) \ge  \pa{Y_{d,V}+t}^{- dK}\sum_{y \in \Chi_{\La_K(\zeta)}} \psi_{g,t}\up{\La}(y),
\eeq
 and hence
\beq\label{dec2}
\pa{\psi_{g,t}\up{\La}(x) }^2\ge \pa{Y_{d,V}+t}^{-2 dK}\sum_{y \in \Chi_{\La_K(\zeta)}} \pa{\psi_{g,t}\up{\La}(y)}^2.
\eeq

Combining \eq{dec1} and \eq{dec2} we get
\beq\label{eq:bndder}
\frac \di {\di t} E\up{\La}(t) \ge Q \pa{Y_{d,V}+t}^{- 2dK}\sum_{\zeta \in  K\Z^d\cap \La}\sum_{y \in \Chi_{\La_K(\zeta)}} \pa{\psi_{g,t}\up{\La}(y)}^2 = Q  \pa{Y_{d,V}+t}^{- 2dK}.
\eeq
Thus
\begin{align}
E\up{\La}(t) - E\up{\La}(0) &\ge Q \int_0^t \di s  \pa{Y_{d,V}+s}^{- 2dK}  \\ \notag & =
\frac {Q}{ 2dK-1}\pa{\frac {1}{Y_{d,V}^{2dK-1}}-  \frac {1}{\pa{Y_{d,V}+t}^{2dK-1}}}.
\end{align}

To conclude the proof of the theorem,  just note that $E_\Gamma(t)=\lim_{R \to \infty} E(\La_R)(t)$ for all $t\ge 0$.
\end{proof}

\subsection{Cheeger's inequality for  the ground state energy}\label{sec:isoperim}  Theorem~\ref{thmEGammaInt} for $H=-\Delta$, namely the lower bound 
 \eqref{EGammaboundInt}, follows from the following theorem.

\begin{theorem} \label{thmEGammaDelta} Let $\Gamma \subsetneq \Z^d$ be  $(K,Q)$-relatively dense.  Then 
\beq\label{EGammabound12}
E_{\Gamma}(-\Delta) \ge \frac 1 {4dK_\ast^{2d}} .
\eeq
In addition,
\beq\label{Etbound12}
 E_\Gamma(-\Delta, t) \ge  \frac 1 {(6d-1)K_\ast^{2d}}\qtx{for} t \ge 2d -1.
\eeq

\end{theorem}

\begin{remark}  For  $H=-\Delta$ the estimate \eq{EGammabound12}   in Theorem~\ref{thmEGammaDelta} is better than the corresponding estimate from Theorem~\ref{thmEGammaV}. Note that \eq{Etbound12} only holds for $ t \ge 2d -1$, giving a lower bound independent of $t$.  We can get an estimate for all $t\ge 0$ by combining \eq{EGammabound12} and \eqref{eq:intermbnd}, getting
\begin{align}\label{Etest3399}
{E}_\Gamma(t) \ge \frac {   t}{4dK_\ast^{2d}\pa{t +4d}\ + 1}.
\end{align} 
This estimate is better than \eq{Etbound12} for sufficiently large $t$.
\end{remark}

Given $A \subset \Z^d$, let
\begin{itemize}
\item     $\partial A= \set{(x,y)\in A \times A^\cc; \abs{x-y}=1}$.

\item  $\partial_- A= \set{x \in A; \; (x,y)\in \partial A \qtx{for some} y \in A^{\cc}}$.

\item  $\partial_+ A= \set{y \in A^{\cc}; \; (x,y)\in \partial A \qtx{for some} x\in A}$.

\item Given $x \in Z^d$, set
\beq
\eta_A(x)= \abs{\set{y \in \Z^d;  (x,y)\in \partial A  }}\in \set{0,1,2,\ldots, 2d},
\eeq
so  $\partial_- A= \set{x\in \Z^d; \; \eta_A(x)\ge 1}$.
\end{itemize}

Note that
\beq
\scal{\Chi_A, (-\Delta )\Chi_A}= \abs{\partial  A}= \sum_{x\in \Z^d} \eta_A(x)= \sum_{x \in \partial_- A}\eta_A(x).
\eeq

\begin{lemma}\label{lembdrybound}
Let $\Gamma \subsetneq \Z^d$ be  $(K,Q)$-relatively dense.  Then for all  $A \subset \Z^d \setminus \Gamma$ we have
\beq\label{partial-}
\scal{\Chi_A,( -\Delta) \Chi_A}=  \abs{\partial  A} \ge K_\ast^{-d} \abs{A}.
\eeq
\end{lemma}

\begin{proof}  Let $A \subset \Z^d \setminus \Gamma_K$, set
$A_\zeta= A \cap \Lambda_K(\zeta)$ for $\zeta \in  K\Z^d$,  and let  $N_A=\abs{\set{\zeta \in  K\Z^d; A_\zeta\ne \emptyset}}$.  Then
\beq\label{NA}
\abs{A} \le K_\ast^{d} N_A .
\eeq
On the other hand, $A_\zeta\ne \emptyset$ implies $ \partial A  \cap \pa{\Lambda_K(\zeta)\ \times \Lambda\up{0}_K(\zeta)}\ne \emptyset$ since $\Gamma_K \cap \Lambda\up{0}_K(\zeta)\ne \emptyset$.  We conclude that
$N_A \le \abs{\partial A}$, so \eq{partial-} follows from \eq{NA}
\end{proof}

Let $H=-\Delta$ and fix  $\Gamma \subsetneq \Z^d$ be  $(K,Q)$-relatively dense.  
Following \cite{LS}, we define the Cheeger constants (note $\scal{\Chi_A,  \Chi_A}=\abs{A}$)
\begin{align}\label{defbetaGamma}
\beta(\Gamma) &= \inf_{A \subset \Z^d\setminus \Gamma; \; 1\le \abs{A} < \infty}  \beta_A(\Gamma), \qtx{where} \beta_A(\Gamma)=\frac{ \scal{\Chi_A, (-\Delta_\Gamma) \Chi_A}}{ \abs{A} },\\
\beta(t) &= \inf_{A \subset \Z^d; \; 1\le \abs{A} < \infty}  \beta_A(t), \qtx{where} \beta_A(t)= \frac{ \scal{\Chi_A, H(t) \Chi_A}}{ \abs{A} } \qtx{and}t\ge  0. \notag
\end{align}
Clearly $\beta(\Gamma)\ge E_{\Gamma}$ and $\beta(t) \ge E(t)$ for all $t \ge 0$.

\begin{lemma}\label{lembeta} We have
\beq
K_\ast^{-d} \le  \beta(\Gamma)  \le 2d.
\eeq
Moreover, $\beta(t)$ is a nondecreasing function of $t\ge 0$, and
\beq
 \beta(t) \ge \beta\up{1}(\Gamma)  \ge K_\ast^{-d} \qtx{for} t \ge 2d -1,
\eeq
where $\beta\up{1}(\Gamma)= \min \set{\beta(\Gamma) ,1}$.
\end{lemma}

\begin{proof}
Given  $A \subset \Z^d\setminus \Gamma$, $\abs{A}\ge 1$, it follows from Lemma~\ref{lembdrybound}, that
\beq	
\beta_A(t)=\beta_A(\Gamma) =\frac{ \scal{\Chi_A, (-\Delta) \Chi_A}}{ \abs{A} }\ge K_\ast^{-d}\qtx{for all} t\ge   0.
\eeq
It follows that  $\beta(\Gamma) \ge K_\ast^{-d}$. On the other hand,
there exists $y_0 \in  \Z^d\setminus \Gamma$, since $\Gamma  \subsetneq \Z^d$,  and we have
\beq
\beta(\Gamma)\le \beta_{\set{y_0}}(\Gamma) \le 2d.
\eeq

Let $A \subset \Z^d; \; 1\le \abs{A} < \infty$.
 Suppose $x \in A\cap \Gamma$,  $A_x=A\setminus \set{x}$, and assume $\abs{A_x} \ge 1$. Then $\abs{A}=\abs{A_x}+1$ and
\beq
\scal{\Chi_A, H(t) \Chi_A}\ge  \scal{\Chi_{A_x}, H(t) \Chi_{A_x}} -2d +t ,    \eeq
so, if $t \ge 2d -1$,
\beq
\beta_{A_x}(t) \le   \frac  {\scal{\Chi_A, H(t) \Chi_A} -1 }{\abs{A}-1 }\le   \frac  {\scal{\Chi_A, H(t) \Chi_A}   }{\abs{A} }=\beta_{A}(t),
\eeq
assumming  $\scal{\Chi_A, H(t) \Chi_A}\le \abs{A}$, i.e., $\beta_A(t)\le 1$.  If $\abs{A\setminus \Gamma}\ge 1$,
repeating this procedure until we removed all points of $\Gamma$ from the set $A$ we obtain
\beq
\beta_{A}(t) \ge  \beta_{A\setminus \Gamma}(t) = \beta_{A\setminus \Gamma}(\Gamma)\ge \beta(\Gamma).
\eeq
If $A\subset \Gamma$, $\abs{A}\ge 1$, we pick $x_0   \in A$, so we get
\beq
\beta_{A}(t) \ge  \beta_{\set{x_0}}(t)  =  2d +t  \ge 2d\ge \beta(\Gamma).
\eeq

We thus conclude that for all $t \ge 2d -1$ we have  $\beta_{A}(t)   \ge \beta\up{1}(\Gamma) $ for all $A \subset \Z^d; \; 1\le \abs{A} < \infty$. The lemma follows.
\end{proof}

Theorem~\ref{thmEGammaDelta} follows from the following theorem.
\begin{theorem} Let $\Gamma \subsetneq \Z^d$ be  $(K,Q)$-relatively dense.  Then\beq\label{EGammabound}
E_{\Gamma}(-\Delta) \ge \frac {\pa{\beta(\Gamma)}^2} {4d} \ge \frac 1 {4dK_\ast^{2d}} .
\eeq
In addition,
\beq\label{Etbound}
E_\Gamma(-\Delta, t)\ge \frac {\pa{\beta\up{1}(\Gamma)}^2} {6d-1} \ge \frac 1 {(6d-1)K_\ast^{2d}}\qtx{for} t \ge 2d -1.
\eeq

\end{theorem}

  \begin{proof} We write $H(t) = H_\Gamma(t)= -\Delta + t \Chi_\Gamma$,  $E_{\Gamma}=E_{\Gamma}(-\Delta)$,  $E( t) =E_\Gamma(-\Delta, t) $.

We prove \eq{Etbound} first.
Following \cite{LS}, we introduce  $\widehat{\Z^d}=\Z^d\cup \set{\infty}$, and for $t>0$  define the self-adjoint bounded operator $\what{H(t)}$ on $\ell^2(\widehat{\Z^d})$ by
\beq
\what{H(t)}\phi(x) = \sum_{y \in \widehat{\Z^d}}   \kappa(x,y)\pa{\phi(x)-\phi(y)},
\eeq
where
\begin{enumerate}
\item  $\kappa(x,y) = 1$ for $x,y \in \Z^d$, $\abs{x-y}=1$,
\item  $\kappa(x,y) = 0$ for $x,y \in \Z^d$, $\abs{x-y}\ne 1$,
\item  $\kappa(x,\infty)=\kappa(\infty,x)= t\Chi_\Gamma(x)$ for $x \in \Z^d$,
\item  $\kappa(\infty,\infty)=0$.
\end{enumerate}

Given $\vphi \in \ell^2(\Z^d)$, we extend it to $\what\vphi \in \ell^2(\what{\Z^d})$ by setting  $\what{\vphi}(\infty)=0$.  It follows that  $\what{H(t)\vphi}= \what{H(t)}\what{\vphi}$, and we have
\begin{align}
\scal {\vphi,H(t)\vphi}_{\ell^2({\Z^d})}=\scal {\what\vphi,\what{H(t)}\what\vphi}_{\ell^2(\what{\Z^d})}= \tfrac 12 \sum_{x,y \in \widehat{\Z^d}}   \kappa(x,y)\abs{\what{\vphi}(x)-\what{\vphi}(y)}^2.
\end{align}

Note that
\beq
\scal {\vphi,H(t)\vphi}=\scal {\vphi,(-\Delta)\vphi} + t  \norm{ \Chi_\Gamma \vphi}^2,
\eeq
so
\begin{align}
E(t)&= \inf \set{\scal {\vphi,H(t)\vphi}: \, \vphi \in \ell^2(\Z^d),\norm{ \vphi} =1 }\\ \notag
&= \inf \set{\scal {\vphi,H(t)\vphi}: \, \vphi \in \ell^2(\Z^d; \R),\norm{ \vphi} =1, \abs{\supp \vphi} < \infty }
\end{align}

Now let $\vphi$ be a real-valued function on $\Z^d$ with finite support.  We have, using the Cauchy-Schwarz inequality,
\begin{align}
2\scal {\vphi,H(t)\vphi}_{\ell^2({\Z^d})}&= \sum_{x,y \in \widehat{\Z^d}}   \kappa(x,y)\pa{\what{\vphi}(x)-\what{\vphi}(y)}^2 \\
& \ge \frac {\pa{\sum_{x,y \in \widehat{\Z^d}}   \kappa(x,y)\abs{\what{\vphi}(x)^2-\what{\vphi}(y)^2}}^2 }{\sum_{x,y \in \widehat{\Z^d}}   \kappa(x,y)\pa{\what{\vphi}(x)+\what{\vphi}(y)}^2 }. \notag
\end{align}

For the denominator, we have
\begin{align}
&\sum_{x,y \in \widehat{\Z^d}}   \kappa(x,y)\pa{\what{\vphi}(x)+\what{\vphi}(y)}^2 =
\sum_{x,y \in {\Z^d}; \, \abs{x-y}=1}\pa{\vphi(x)+\vphi(y)}^2 +2 t \scal{\vphi, \Chi_\Gamma \vphi}\\
& \quad \le \sum_{x,y \in {\Z^d}; \, \abs{x-y}=1}\pa{2\vphi(x)^2+2\vphi(y)^2} +2 t \scal{\vphi, \Chi_\Gamma \vphi} \le 8d \norm{\vphi}^2 +2 t \norm{ \Chi_\Gamma \vphi}^2.
\notag
\end{align}
For the numerator, since $\kappa$ is symmetric, we have, setting  
\[A_s=\set{\what{\vphi}^2 >s}= \set{{\vphi}^2 >s}\] for $s \ge 0$,
\begin{align}\notag
&\sum_{x,y \in \widehat{\Z^d}}   \kappa(x,y)\abs{\what{\vphi}(x)^2-\what{\vphi}(y)^2}\\ \notag
&\qquad =
2 \sum_{x,y \in \widehat{\Z^d}}   \kappa(x,y)\Chi\pa{\set{\what{\vphi}(x)^2>\what{\vphi}(y)^2}}\abs{\what{\vphi}(x)^2-\what{\vphi}(y)^2}\\
&\qquad  =2\int_0^\infty \di s \sum_{x,y \in \widehat{\Z^d}}   \kappa(x,y)\Chi\pa{\set{\what{\vphi}(x)^2>s \ge\what{\vphi}(y)^2}}\\
&\qquad  = 2\int_0^\infty \di s \sum_{x,y \in \widehat{\Z^d}}   \kappa(x,y)\Chi_{A_s}(x) \pa{\Chi_{A_s}(x) -\Chi_{A_s}(y)} \notag\\
&\qquad = 2 \int_0^\infty \di s\scal {\Chi_{A_s},H(t)\Chi_{A_s}}
 \ge    2  \beta(t)  \int_0^\infty \di s \abs{A_s} =   2  \beta(t)  \norm{\vphi}^2 .\notag
\end{align}

We conclude that for  a real-valued function $\vphi$  on $\Z^d$ with finite support and $\norm{\vphi}=1$  we have, for all $t \ge 2d -1$,  using Lemma~\ref{lembeta},
\begin{align}
\scal {\vphi,H(t)\vphi} \ge \tfrac 12 \frac{ \pa{2  \beta(t)  \norm{\vphi}^2}^2}{8d \norm{\vphi}^2 +2 t \norm{ \Chi_\Gamma \vphi}^2}\ge \frac{ \pa{  \beta(t)  }^2}{4d  + t }\ge
\frac{ \pa{  \beta\up{1}(\Gamma)  }^2}{4d  + t }.
\end{align}
Thus \beq
E(t) \ge  \frac{ \pa{  \beta\up{1}(\Gamma)  }^2}{4d  + t } \qtx{for} t \ge 2d -1.
\eeq
  Since $E(t)$ is nondecreasing in $t$, we get
  \beq
E(t) \ge  \frac{ \pa{  \beta\up{1}(\Gamma)  }^2}{6d -1 }\qtx{for all} t \ge 2d -1.
\eeq

To prove \eq{EGammabound}, we repeat the above procedure with $-\Delta_\Gamma$, $\Z^d\setminus \Gamma$, $\Z^d$ and $\what{-\Delta_\Gamma}$ instead of $H(t)$, $\Z^d$,
$\what{\Z^d}$ and $\what{H(t)}$, where
$\what{-\Delta_\Gamma}= (-\Delta_\Gamma) \oplus  0  \qtx{on} \ell^2(\Z^d)= \ell^2(\Z^d\setminus \Gamma)  \oplus \ell^2(\Gamma)$,
and $\kappa(x,y) = 1$ for $x,y \in \Z^d$, $\abs{x-y}=1$,
$\kappa(x,y) = 0$ for $x,y \in \Z^d$, $\abs{x-y}\ne 1$, and, given $\vphi \in \ell^2(\Z^d\setminus \Gamma)$, extending it to $\what\vphi\in \ell^2(\Z^d)$ by setting  $\what\vphi(x)=0$ for $x \in \Gamma$. The proof goes through in exactly the same way, and we get \eq{EGammabound}.
\end{proof}

\section{Trimmed Anderson models}

In this section we prove Proposition \ref{prop:admfam} and Theorem \ref{thmWegbottom}.

\subsection{The ground state energy}\label{secbottomsp}

\begin{proof}[Proof of Proposition \ref{prop:admfam}] Let $H_{\bom,\lambda}$ be a $\Gamma$-trimmed  Anderson model with  $\mu_\zeta=\mu$ for all $\zeta \in \Gamma$ with $0=\inf \supp \mu$. 
To show that $E_\emptyset=E_\emptyset(H_0)\in\sigma_{{\rm ess}}(H_0)$ we construct an orthonormal sequence $\set{\phi_n}_{n\in \N}$  in $\ell^2(\Z^d)$ such that 
\beq
\|(H_0-E_\emptyset)\phi_n\|\le 1/n \qtx{for all} n \in \N
\eeq
The existence of such sequence is readily guaranteed by \eqref{eq:condprob}. Hence  $E_\emptyset\in\sigma_{{\rm ess}}(H_0)$ by Weyl's criterion.

To show that \eq{E0prob} holds, for each $\eps>0$ we use \eqref{eq:condprob} to construct   an orthonormal sequence $\set{\psi\up{\epsilon}_n}_{n\in \N}$  in $\ell^2(\Z^d)$ such that
$\supp \psi\up{\epsilon}_n\subset \Lambda_L(x_n)$ with $L=L(\epsilon)$ for all $n\in N$,
with $\norm{x_n-x_M}_\infty >L$ for $n\ne m$, and 
\begin{gather}
\|(H_0-E_\emptyset)\psi\up{\epsilon}_n\|\le \epsilon \qtx{for all} n \in \N.
\end{gather}

We then have 
\beq
E_\emptyset(H_{\bom,\lambda})\le \scal{\psi\up{\epsilon}_n, H_{\bom,\lambda}\psi\up{\epsilon}_n}\le \epsilon+\sum_{\zeta \in \Gamma\cap \Lambda_L(x_n)} \omega_\zeta\abs{\psi\up{\epsilon}_n(\zeta)}^2 \sqtx{for all} n\in\N.
\eeq
 But 
\beq
\P \set{\inf_{n\in \N}\max_{\zeta \in \Lambda_L(x_n)} \omega_\zeta<\epsilon L^{-d}}=1,
\eeq
from which it follows that 
\beq
\P \set{\sigma(H_{\bom,\lambda}) \cap [E_\emptyset,E_\emptyset+2\epsilon]\neq\emptyset)}=1.
\eeq 
Since $\epsilon$ is arbitrary, the result follows.
\end{proof}

\subsection{The Wegner estimate}\label{secwegner}

\begin{proof}[Proof of Theorem \ref{thmWegbottom}]

Let $H_{\bom,\lambda}$ be a $\Gamma$-trimmed  Anderson model,  fix    $E_{1}\in (E_\emptyset(H_0), E_\Gamma(H_0))$, and let  $\kappa = \kappa(H_0,\Gamma,E_1)$ be  as in 
\eq{kappabottom}.  We clearly have $\kappa>0$.

(We can derive a lower bound for  $\kappa$, as stated in Remark~\ref{remkappa}. The estimate  \eqref{delgambndInt}  states that
\beq\label{kappa1}
 E_\Gamma(H_0,s)-E_\emptyset(H_0) \ge   \frac Q{\pa{ 2dK-1}{Y_{d,V\up{0}}^{2dK-1}}}\pa{1-\pa{ \frac {Y_{d,V\up{0}}}{Y_{d,V\up{0}}+s}}^{2dK-1}}
 \eeq
for all $s>0$,  which implies   $E_\Gamma(H_0,s) >E_1$ for 
\beq\label{kappa2}
s > s_0=  Y_{d,V\up{0}}\pa{ \pa{1-   { (E_1-E_\emptyset(H_0))Q^{-1}(2dK-1)Y_{d,V\up{0}}^{2dK-1}} }^{-\frac 1 { 2dK-1}}   -1}.
\eeq
Using \eqref{eq:bndder}, we get
 \begin{align}\label{kappa3}
\kappa \ge \sup_{s>s_0} \frac {E_\Gamma(H_0,s) -E_\Gamma(H_0,s_0)}s   \ge \sup_{s>s_0} \frac {s-s_0}s  Q\pa{Y_{d,V\up{0}}+s}^{-2dK}.
\end{align}
The supremum  is attained at
\beq\label{kappa4}
s= \frac {2Kd+1}{4Kd} \pa{1 + \sqrt{1+ \frac {8Kd Y_{d,V\up{0}}}{(2Kd+1)^2s_0   } } }s_0;
\eeq
to get a simpler lower bound we take $s=s= \frac {2Kd+1}{2Kd} s_0$,  getting
\beq\label{kappa5}
\kappa \ge  \frac{Q} {2dK+1} \pa{ Y_{d,V\up{0}}+  \frac {2Kd+1}{2Kd} s_0}^{-2dK},
\eeq
which is  \eqref{eq:bndkappa}.)

We now proceed as in \cite[Proof of Theorem~1.7]{Kuc}.  Let $\La=\La_L(x_0)$ with $x_0 \in \Z^d$ and $L>0$, and note that  ($H\up{\La}_{0,\Gamma}(t)=  \pa{(H_0)_\Gamma (t)}\up{\La} $)
\beq
E_\Gamma(H^\La_0,t)=E_\emptyset \pa{H\up{\La}_{0,\Gamma}(t)}\ge E_\emptyset \pa{H_{0,\Gamma}(t)}=E_\Gamma(H_0,t),
\eeq
so
\beq
 \kappa(H^\La_0,\Gamma,E_1)= \sup_{s>0; \; E_\Gamma(H^\La_0,s) >E_1} \frac {E_\Gamma(H^\La_0,s) -E_{1}}s \ge \kappa(H_0,\Gamma,E_1)=\kappa>0.
\eeq
As a consequence, \eq{PVPb} follows immediately from \cite[Lemma~4.1]{Kuc}.

The Wegner estimate \eq{preWegner} follows using \eq{PVPb}. For any closed interval $I \subset (-\infty,E_{1}]$ we have
\begin{align}
 \tr \Chi_{I}(H_{\bom,\lambda}\up{\La})& \le \kappa^{-1} \tr \Chi_{I}(H_{\bom,\lambda}\up{\La})\Chi_{\Gamma \cap \La}\Chi_{I}(H_{\bom,\lambda}\up{\La})= \kappa^{-1} \tr \Chi_{\Gamma \cap \La}\Chi_{I}(H_{\bom,\lambda}\up{\La})\Chi_{\Gamma \cap \La}\\ \notag  &
 = \kappa^{-1} \sum_{\zeta \in \Gamma \cap \La} \scal{\delta_\zeta, \Chi_{I}(H_{\bom,\lambda}\up{\La})\delta_\zeta }.
\end{align}
Since by spectral averaging \cite[Eq.~(3.16)]{CHK2} (see also \cite[Appendix~A]{CGK})
\beq
\int \di \mu_\zeta (\omega_\zeta) \scal{\delta_\zeta, \Chi_{I}(H_{\bom,\lambda}\up{\La})\delta_\zeta } \le 8S_{ \mu_\zeta}(\lambda^{-1}\abs{I}),
\eeq
we get  \eq{preWegner}. 
\end{proof}

\begin{acknowledgement}
We are grateful to Sasha Sodin for useful discussions on isoperimetric estimates.
\end{acknowledgement}

\end{document}